\documentclass[11pt]{amsart}

\usepackage{amsmath, amssymb, amsthm, amscd}
\usepackage{url}
\usepackage{graphicx}
\usepackage[hidelinks,pagebackref,pdftex]{hyperref}
\usepackage{color}
\usepackage{import}
\usepackage[margin=3cm]{geometry}

\usepackage{enumerate}
\usepackage[linesnumbered]{algorithm2e}
\usepackage{amsfonts}

 { \begin{list}%
         {$\bullet$}%
         {\setlength{\labelwidth}{20pt}%
          \setlength{\leftmargin}{25pt}%
          \setlength{\topsep}{0pt}
          \setlength{\itemsep}{0pt}
          \setlength{\parsep}{0pt}}}%
 { \end{list} }

\newcommand{\repseq}[3]{\xymatrix @-1pc{#1^{(#2)} \ar@{<->}[r] & \cdots \ar@{<->}[r] & #1^{(#3)} }}

\newcommand{\Ring}{\mathcal{R}}

\newcommand{\Field}{\mathbb{F}}

\newcommand{\Z}{\mathbb{Z}}
\newcommand{\R}{\mathbb{R}}
\newcommand{\Q}{\mathbb{Q}}

\newcommand{\X}{\mathbb{X}}

\newcommand{\modulo}{ \ \mathrm{mod}\ }

\newcommand{\low}{\mathrm{low}}

\newcommand{\Rmat}{\mathbf{R}}
\newcommand{\Matrix}{\mathbf{M}}
\newcommand{\col}{\mathsf{col}}

\newcommand{\V}{V}                        
\newcommand{\K}{\mathbf{K}}               

\newcommand{\Chain}{\mathbf{C}}

\newcommand{\Hom}{\mathbf{H}}
\newcommand{\Cycles}{\mathbf{Z}}
\newcommand{\Boundaries}{\mathbf{B}}

\newcommand{\wordlength}{\mathsf{w}} 
\newcommand{\Acpx}{\mathsf{A}}       

\newcommand{\bits}{\mathsf{B}}

\usepackage{pgfplots}
\usetikzlibrary{pgfplots.statistics}
\pgfplotsset{width=7cm,compat=1.14}

\numberwithin{equation}{section}
\theoremstyle{plain}
\newtheorem{theorem}[equation]{Theorem}

\newtheorem{lemma}[equation]{Lemma}

\newtheorem{corollary}[equation]{Corollary}
\newtheorem{proposition}[equation]{Proposition}
\newtheorem{definition}[equation]{Definition}
\newtheorem{conjecture}[equation]{Conjecture}

\begin{document}
\title{Computing Persistent Homology with Various Coefficient Fields in a Single Pass}

\author{Jean-Daniel Boissonnat} 
\author{Cl\'ement Maria}

\address{INRIA Sophia Antipolis-M\'editerran\'ee}
\email{jean-daniel.boisonnat@inria.fr, clement.maria@inria.fr}


\begin{abstract}
This article introduces an algorithm to compute the persistent homology of a filtered 
complex with various coefficient fields in a single matrix reduction. The algorithm is output-sensitive in 
the total number of \emph{distinct} persistent homological features in the diagrams for the different 
coefficient fields. This computation allows us to infer the prime divisors 
of the torsion coefficients of the integral homology groups of the topological space at any scale, 
hence furnishing a more informative description of topology than persistence in a single coefficient field. 
We provide theoretical complexity analysis as well as detailed experimental results. The code is 
part of the {\tt Gudhi} software library, and is available at~\cite{gudhi:PersistentCohomology}.

\bigskip

This article appeared in the Journal of Applied and Computational Topology 2019~\cite{DBLP:journals/jact/BoissonnatM19}. An extended abstract of this article appeared in the proceedings of the European Symposium on Algorithms 2014~\cite{DBLP:conf/esa/BoissonnatM14}.

\end{abstract}

\maketitle


\section{Introduction}
\label{sec:mf_intro}
Persistent homology~\cite{DBLP:books/daglib/0025666,DBLP:journals/dcg/ZomorodianC05} is an invariant measuring the topological features of the sublevel sets
of a function defined on a topological space. Its generality and stability~\cite{DBLP:journals/dcg/Cohen-SteinerEH07} with regard to noise have made it a widely used tool in applied topology. When considering homology with field coefficients---in opposition to integer coefficients---persistent homology admits an algebraic decomposition that can be represented by a \emph{persistence diagram}~\cite{DBLP:journals/dcg/ZomorodianC05}. 
The persistence diagram contains rich information about the topology of the studied space and very efficient methods exist to compute it. 
However, the integral homology groups of a topological space are strictly more informative than the homology groups with field coefficients, in particular because they convey information about \emph{torsion} in homology. Algebraically, torsion is characterized by cyclic subgroups of the integral homology groups, and appears in the range of application of computational topology, such as topological data analysis---where, for example, Klein bottles appear naturally~\cite{DBLP:journals/ijcv/CarlssonISZ08,martin2010top}---or the study of random complexes---where a burst of torsion subgroups of large order are found~\cite{2017arXiv171005683K}.

When homology is computed with field coefficients, these torsion subgroups may either vanish or contribute to the homology, depending on their (unknown) orders. This consequently obfuscate the study of the topology of data and complexes. A simple approach to distinguish between the two cases is to compute persistent homology with different coefficient fields and track the differences in the persistence diagrams. 

We build on this idea and describe an efficient algorithm to compute persistent homology with various coefficient fields $\Z/{q_1}\Z, \cdots, \Z/{q_r}\Z$ in a single pass of the matrix reduction algorithm. 
To do so, we introduce a method we call \emph{modular reconstruction} consisting of using the \emph{Chinese Remainder Isomorphism} to encode an element of $\Z/{q_1}\Z \times \cdots \times \Z/{q_r}\Z$ with an element of $\Z/{(q_1 \cdots q_r)}\Z$. This is a simple solution to implement a simple idea. However, it requires the introduction of technical tools for dedicated arithmetic operations, and the solution is tailored for persistent homology computations. 

Specifically, we describe algorithms to perform elementary row/column operations in a matrix with $\Z/{(q_1 \cdots q_r)}\Z$ coefficients, corresponding to simultaneous elementary row/column operations in $r$ distinct matrices with coefficients in the fields $\Z/{q_1}\Z, \cdots, \Z/{q_r}\Z$ respectively. The method results in an algorithm with an output-sensitive complexity in the total number of \emph{distinct} pairs in the echelon forms of the matrices with $\Z/{q_1}\Z, \cdots, \Z/{q_r}\Z$ coefficients, plus an overhead due to arithmetic operations on big numbers in $\Z/{(q_1 \cdots q_r)}\Z$. We present the method for computing persistent homology with several coefficient fields using the original persistence algorithm~\cite{DBLP:journals/dcg/EdelsbrunnerLZ02,DBLP:journals/dcg/ZomorodianC05}, but the methodology and generic tools developed may be applied to other persistent homology algorithms relying on elementary row/column operations, such as the persistent cohomology algorithms of~\cite{DBLP:journals/algorithmica/BoissonnatDM15,DBLP:journals/dcg/SilvaMV11,DBLP:conf/compgeom/DeyFW14}. Finally, we describe how to infer the torsion coefficients of the integral homology using the \emph{Universal Coefficient Theorem for Homology}, and how to integrate this information in a \emph{multi-field persistence diagram}, that could be used in application pipelines.

We discuss applications of the algorithm, and provide experimental analysis that on practical examples of interest, our method is significantly faster than the brute-force approach consisting in reducing separately $r$ matrices with coefficients in $\Z/{q_1}\Z, \cdots , \Z/{q_r}\Z$. It is important to note that the method does not pretend to scale to large $r$, as the arithmetic complexity of operations in $\Z/{(q_1 \cdots q_r)}\Z$ becomes problematic. 

Computing persistent homology with different coefficients has been mentioned in the literature~\cite{DBLP:journals/dcg/ZomorodianC05} in order to verify if a persisting feature was due to an actual ``hole'' (or high-dimensional equivalent) or to torsion (and consequently existed only for a certain coefficient field). The issues caused by homological torsion in the study of data using persistent homology is also discussed in~\cite{DBLP:journals/dcg/SilvaMV11}. 
To the best of our knowledge, this is the first work describing an efficient and practical algorithm to compute persistence with various coefficient fields in order to detect and analyse torsion subgroups in persistent homology.

\section{Background}
\label{sec:background}

For simplicity, we focus in the following on simplicial complexes and their homology. However, the approach and the algorithms do not rely on the simplicial structure, and apply to general complexes.

\subsection{Simplicial Homology with General Coefficients} 

We refer the reader to~\cite{Hatcher-algebraictopology2001} for an introduction to
homology and to~\cite{DBLP:books/daglib/0025666} for an introduction to persistent homology.

A {\em simplicial complex} $\K$ on a set of \emph{vertices} $\V = \{1, \cdots ,n\}$ is a collection of simplices $\{\sigma\}$, $\sigma \subseteq V$, such that $\tau \subseteq \sigma \in \K \Rightarrow \tau \in \K$. The dimension $d=|\sigma|-1$ of $\sigma$ is its number of elements minus $1$. For a ring $\Ring$, the group of $d$-chains, denoted by $\Chain_d(\K,\Ring)$, of $\K$ is the group of formal sums of $d$-simplices with $\Ring$ coefficients. The \emph{boundary operator} is a linear operator $\partial_d: \Chain_d(\K,\Ring) \rightarrow \Chain_{d-1}(\K,\Ring)$ such that $\partial_d \sigma =
\partial_d [v_0, \cdots , v_d] = \sum_{i=0}^d (-1)^{i}[v_0,\cdots ,\widehat{v_i}, \cdots,v_d]$, where $\widehat{v_i}$ means $v_i$ is deleted from the list. It will be convenient to consider later the  endomorphism  $\partial_*:\bigoplus_d \Chain_d(\K,\Ring) \rightarrow \bigoplus_d
\Chain_d(\K,\Ring)$ extended by  linearity to the external sum of  chain groups. Denote by $\Cycles_d(\K,\Ring)$ and $\Boundaries_{d-1}(\K,\Ring)$ the kernel and the image of $\partial_d$ respectively. Observing $\partial_d \circ \partial_{d+1}=0$, we define the $d^{th}$ homology group $\Hom_d(\K,\Ring)$ of $\K$ by the quotient $\Hom_d(\K,\Ring) = \Cycles_d(\K,\Ring)/\Boundaries_d(\K,\Ring)$. 

If $\Ring$ is the \emph{ring of integers} $\Z$, $\Hom_d(\K,\Z)$ is an abelian group and, according to the 
\emph{fundamental theorem of finitely generated abelian groups}~\cite{Hatcher-algebraictopology2001}, admits a 
\emph{primary decomposition}: 
\begin{equation}
\label{eq:group_decomposition}
\Hom_d(\K,\Z) \cong \Z^{\beta_d(\Z)} \bigoplus_{q \text{ prime}}
\left( \Z/{q^{k_1}}\Z \oplus \cdots \oplus
  \Z/{q^{k_{t(d,q)}}}\Z  \right)
\end{equation}
\noindent
for a uniquely defined integer $\beta_d(\Z)$, called the \emph{$d^{\text{th}}$ integral Betti number}, and integers $t(d,q) \geq 0$ and $k_i > 0$  for every prime number $q$. If $t(d,q)>0$, the integers $q^{k_1},\cdots,q^{k_{t(d,q)}}$ are called \emph{torsion coefficients}, and they admit $q$ as unique \emph{prime divisor}. Intuitively, in dimension $0$, $1$ and $2$, the integral Betti numbers count the number of connected components, the number of holes and the number of voids respectively. Torsion captures features such as non-orientability in surfaces ; see Section~\ref{app:mf_persistencediagram} and Figure~\ref{fig:diagram} for the example of the Klein bottle.  
If $\Ring$ is a \emph{field} $\Field$, $\Hom_d(\K,\Field)$ is a vector-space and decomposes into 
\[
\Hom_d(\K,\Field) \cong \Field^{\beta_d(\Field)}
\]
\noindent
where $\beta_d(\Field)$ is the \emph{$d^{\text{th}}$ field Betti number}.
The field Betti numbers $\left(\beta_d(\Field)\right)_d$ are entirely determined by the characteristic 
of $\Field$ and the integral homology ; see Section~\ref{sec:mf_uct}. Hence, the integral homology is more informative than homology in $\Field$. 

We suggest in Section~\ref{sec:mf_expe} the study of the $\Z$-homology of geometric data and random complexes. It is unclear how often integral homology is more informative that field homology in general geometric data, but important cases where torsion is fundamental in the study of data have been observed~\cite{DBLP:journals/ijcv/CarlssonISZ08,martin2010top}. The analysis of torsion is however fundamental in the study of random complexes~\cite{Luczak2018}.

\subsection{Persistent Homology with Field Coefficients} 

A \emph{filtration} of a complex is a function $f:\K \rightarrow \R$ satisfying $f(\tau) \leq f(\sigma)$ whenever $\tau \subseteq \sigma$. Ordering the simplices of $\K$ by strictly increasing $f$-value, we get an increasing sequence of complexes
\[
\emptyset = \K_0 \subsetneq \K_1 \subsetneq \cdots \subsetneq \K_{m-1}
\subsetneq \K_m = \K
\]
where all simplices in $\K_i \setminus \K_{i-1}$ have same filtration value. Without loss of generality, we suppose in the following that all $f$-values are distinct, and that successive complexes differ by exactly one simplex, i.e., $\K_i = \K_{i-1} \cup \{\sigma_i\}$. The {\em size} of a filtration is the number of simplices $m$ in the complex $\K$. 

A filtration induces a sequence of $d$-homology groups 
\[
0 = \Hom_d(\K_0,\Ring) \rightarrow \Hom_d(\K_1,\Ring) \rightarrow \cdots \rightarrow
\Hom_d(\K_m,\Ring) = \Hom_d(\K,\Ring)
\]
\noindent
connected by homomorphisms, induced by the inclusions. In the following, we denote simply by $\K$ the filtration $(\K,f)$. When $\Ring$ is a field, the latter sequence admits an algebraic decomposition that can be described in terms of a family of intervals $\{(i,j)\}$, called an \emph{indexed persistence diagram}, where a pair $(i,j)$ belongs to $\{1 \ldots m\} \times \{1 \ldots, m , \infty\}$ and is interpreted as a homology feature that \emph{is born} at index $i$ and \emph{dies} at index $j$ (homology features which never die have death $\infty$). Note than, for simplicial complexes, an index $i \in \{1 \ldots n\}$ belongs to exactly one pair (as birth or death) of the indexed persistence diagram. We assume this property true in the remainder of the article. 
For a fixed field of coefficients $\Field$, computing the persistent homology of a filtration consists of computing the persistence diagram of the induced sequence of $\Field$-homology groups. 
%

\section{Multi-Field Persistent Homology}
\label{sec:mf_motivation}

We call the algorithmic problem of computing persistent homology for a family of coefficient fields $\Z/q_1\Z, \ldots, \Z/q_r\Z$ \emph{multi-field persistent homology}. As explained in the next section, computing multi-field persistence allows us to infer a more informative description of the topology of a space, compared to persistence in a single field.

\medskip

\subsection{Inference of Torsion}
\label{sec:mf_uct}

For a topological space $\X$, the \emph{Universal Coefficient Theorem for Homology}~\cite{Hatcher-algebraictopology2001} establishes the relationship between the homology groups $\Hom_d(\X,\Z)$ with $\Z$ coefficients and the
homology groups $\Hom_d(\X,\Z/q\Z)$ with coefficients in the field $\Z/q\Z$ 
(of characteristic $q$), for $q$ prime. We use the following corollary:
\begin{corollary}[Universal Coefficient Theorem~\cite{Hatcher-algebraictopology2001}{[Corollary 3A.6.(b)]}]
  \label{cor:uct}
  Denote by $\beta_d(\Z)$ and $\beta_d(\Z/q\Z)$ the Betti numbers of
  $\Hom_d(\X,\Z)$ and $\Hom_d(\X,\Z/q\Z)$ respectively, and $t(j,q)$ the number of $\Z/{q^{k_i}}\Z$ summands in the primary decomposition of the homology group $\Hom_j(\X,\Z)$ as in Equation~(\ref{eq:group_decomposition}), we have:
  $$\beta_d(\Z/q\Z) = \beta_d(\Z)+t(d,q) + t(d-1,q)$$
\end{corollary}

Suppose $\{q_1, \cdots ,q_r\}$ are the first $r$ prime
numbers and $q_r$ is a strict upper bound on the prime divisors of the
torsion coefficients of $\X$. Consequently, according to Corollary~\ref{cor:uct}, $\beta_d(\Z/q_r\Z) = \beta_d(\Z)$ for all dimensions $d$. Moreover, there is no torsion in $0$-homology~\cite{Hatcher-algebraictopology2001}, and 
$t(0,q)=0$ for all primes $q$. Given the Betti numbers of $\X$ in all
fields $\Z/q_s\Z, 1\leq s \leq r$, we deduce from Corollary~\ref{cor:uct} the
recurrence formula $t(d,q_s) = \beta_d(\Z/q_s\Z) - \beta_d(\Z/q_r\Z) - t(d-1,q_s)$, from which we compute the value of $t(d,q)$ for every dimension $d$ and prime $q$. For any dimension $d$, we consequently infer the integral Betti numbers and the number $t(d,q)$ of $\Z/{q^{k_i}}\Z$ summands in the primary decomposition of $\Hom_d(\X,\Z)$. 

\medskip

It is important to notice two limitations of this approach. First, the universal coefficient theorem does not allow us to infer powers $k_i$ from the summands $\Z/q_i^{k_i} \Z$ in the decomposition of the homology groups with $\Z$-coefficient, as in Equation~(\ref{eq:group_decomposition}), by computing homology with field coefficients. Consequently, a summand $\Z/q^{k_i} \Z$ is detected as a summand $\Z/q^{*} \Z$, for an unknown power of $q$.
Second, determining an upper bound $q_r$ on the prime divisors of the torsion coefficients of a complex is a difficult task in general. However, computing separately persistent homology with $\Q$-coefficients provides the Betti numbers $\beta_d(\Q)$ that are equal to $\beta_d(\Z)$, and can be used in the formula of Corollary~\ref{cor:uct}. This allows us to detect correctly the summands $\Z/q^{k_i} \Z$ for all $q \leq q_r$, even when $q_r$ is not an upper bound on the prime divisors of the torsion coefficients.

We discuss the question of upper bounds of prime divisors of torsion coefficients in the experimental Section~\ref{sec:mf_expe} for different types of data sets.

\subsection{Representation of the Multi-Field Persistence Diagram} 
\label{app:mf_persistencediagram}

Persistence diagrams are represented by sets of points in the plane, where to every persistent pair $(i,j)$ of the diagram corresponds a point with coordinates $(i,j)$ in the plane ; see Figure~\ref{fig:diagram}. We generalize this representation to multi-field persistence diagram by plotting the \emph{superimposition} of the persistence diagrams in each coefficient field, and by infering an expression of the integral homology group in each cell of the diagram. 

We refer to Figure~\ref{fig:diagram} for an example. It pictures the multi-field persistence diagram of the $1$-homology of a filtration $\K$ approximating a Klein bottle (for field coefficients $\Z/2\Z$ and $\Z/3\Z$). The integral $1$-homology of the Klein bottle is $\Hom_1(\K,\Z) = \Z \oplus \Z/2\Z$, and $\Hom_1(\K,\Z/2\Z) = (\Z/2\Z)^2$ and $\Hom_1(\K,\Z/3\Z) = \Z/3\Z$, and the integral homology appears clearly in the multi-field persistence diagram.

Notion of distances, such as {\em bottleneck distance} and \emph{Wasserstein distance}, and stability~\cite{DBLP:journals/dcg/Cohen-SteinerEH07}, extend naturally to this presentation, by defining the distance between two multi-field persistence diagram for coefficients $\Z/q_1, \ldots , \Z/q_r\Z$ as the maximal distance between the corresponding standard persistence diagrams over all coefficient fields $\Z/q_s\Z$, $1 \leq s \leq r$.

\begin{figure}[t!]
\centering
  \includegraphics[width=6cm]{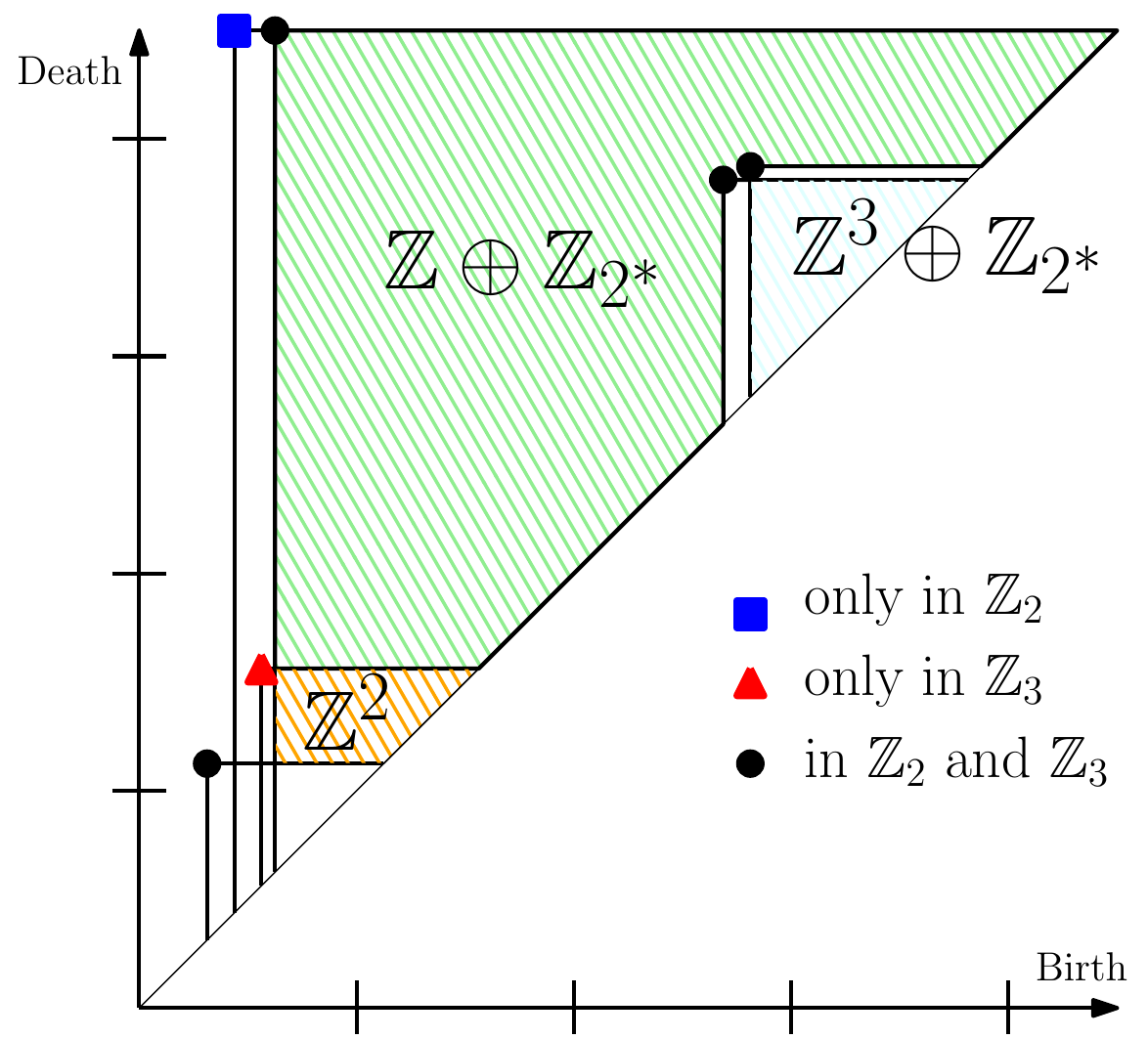}
   \caption{Multi-field persistence
     diagram of the most persistent features of $\Hom_1$ for a Rips complex
     reconstructing a Klein bottle. The ``$*$'' in $\Z/{2^*}\Z$ indicates that the persisting homology admits a torsion summand $\Z/{2^k}\Z$ for some unknown $k \geq 1$. The $1$-homology $\Z \oplus \Z/2\Z$ of the underlying Klein bottle appears clearly as persisting.}
  \label{fig:diagram}
\end{figure}

\section{Algorithm for Multi-Field Persistent Homology}
\label{sec:mf_algo}

In this section we design an efficient algorithm to compute multi-field persistent homology. For a filtered complex $\K$ of size $m$, denote by $P_{\Field}$ the number of pairs of indices $(i,j) \in \{1, \ldots , m\} \times \{1, \ldots , m, \infty\}$ forming the index persistence diagram with coefficients in a field $\Field$. For a set of coefficient fields $\Z/q_1\Z, \ldots, \Z/q_r\Z$, denote by $P_r$ the number of \emph{distinct} pairs of indices appearing the persistence diagram of $\K$ for every coefficient field $\Z/q_s\Z$.

We design an algorithm, called \emph{modular reconstruction} algorithm, of complexity
\[
O\left( \left[ r \times (P_r-P_\Field) + P_r^3 \right] \times \Acpx_r\right)
\]
where $\Acpx_r$ is a bound on the time complexity of arithmetic operations on large integers in $\Z/(q_1 \cdots q_r)\Z$ (see Section~\ref{sec:arithcpx}), and the $P_r^3$ stands for the standard cubic complexity of computing persistent homology. Note that the additional component $r \times (P_r-P_\Field)$ depends on the number of distinct bars in the persistence diagram when changing coefficient fields which, in light of Section~\ref{sec:mf_uct}, is directly related to torsion. In that sense, the algorithm is output-sensitive.

\smallskip

For clarity, we focus in this section on the persistent homology algorithm as presented in~\cite{DBLP:books/daglib/0025666}[Chapter VII], which consists of a reduction to column echelon form (defined later) of a matrix. All practical persistent homology algorithms rely on atomic matrix column operations. Our approach to multi-field persistence is described in terms of these column operations, and can consequently be adapted to other practical persistent homology implementations. In the following, $\Z/n\Z$ denotes the ring $(\Z/n\Z,+,\times)$ for any integer $n \geq 1$ and $(\Z/n\Z)^\times$ the subset of invertible elements for $\times$. 
If it exists, we denote the inverse of $x \in \Z/n\Z$ by $x^{-1}$.

\subsection{Persistent Homology Algorithm} 

In this section we recall the standard matrix algorithm to compute persistent homology with coefficient in a field~\cite{DBLP:books/daglib/0025666}[Chapter VII]. 

\smallskip

For an $m \times m$ matrix $\Matrix$, denote by $\col_j$ the 
$j^{\text{th}}$ column of $\Matrix$, $1 \leq j \leq m$, and denote by $\col_j[k]$ the $k^{\text{th}}$ entry of the column. Let $\low(j)$ denote the row index of the lowest non-zero entry of $\col_j$. If the column $j$ is entirely zero, $\low(j)$ is undefined. We say that $\Matrix$ is in \emph{reduced column echelon form} if, for any two non-zero columns $\col_j$ and $\col_{j'}$, $j \neq j'$, the columns satisfy $\low(j) \neq \low(j')$. 

Let $\K = (\sigma_i)_{i=1 \cdots m}$ be a filtered complex. For a fixed coefficient field $\Field$, its boundary matrix $\Matrix_\partial$ is the $m \times m$ matrix, with $\Field$ entries, of the endomorphism $\partial_{*}$ in the basis $\{\sigma_1, \cdots , \sigma_m\}$ of $\bigoplus_d \Chain_d(\K,\Field)$. The basis is ordered according to the filtration. It is a matrix with $\{-1,0,1\}$ entries, where $0$ and $1$ denote the identity $0_{\Field}$ for $+$ and the identity $1_{\Field}$ for $\times$ in $\Field$ respectively, and $-1$ is the inverse of $1_{\Field}$ in $\Field$.
The persistent homology algorithm consists of a left-to-right reduction to column echelon form of $\Matrix_\partial$, presented in Algorithm~\ref{algo:ph}. We denote by $\Rmat$ the matrix we reduce, with columns $\col_j$, and which is initially equal to $\Matrix_\partial$. The algorithm returns the \emph{(indexed) persistence diagram}, which is the set of pairs $\{(\low(j),j)\}$ in the reduced column echelon form of the matrix. Note that the ``infinite intervals'' of the diagram can be infered by reading the null columns of the reduced matrix, and for simplicity we do not include this computation in the pseudo-code.

\begin{algorithm}[t]
\SetAlgoLined
\KwData{Boundary matrix $\Rmat \leftarrow \Matrix_\partial$, persistence diagram $\mathcal{P}_\Field \leftarrow \emptyset$}
\KwOut{Persistence diagram $\mathcal{P}_\Field = \{(i,j)\}$}
\For{$j=1, \cdots ,m$}{
    \While{there exists $j' < j$ with $\low(j') = \low(j)$}{
        $k \leftarrow \low(j)$\;
        $\col_j \leftarrow \col_j - \left(\col_j[k] \times \col_{j'}[k]^{-1} \right) \cdot \col_{j'}$\;
    }
    \lIf{$\col_j \neq 0$}{$\mathcal{P}_\Field \leftarrow \mathcal{P}_\Field \cup \{(\low(j),j)\}$}
}
\caption{Persistent homology algorithm.}
\label{algo:ph}
\end{algorithm}

The reduced form of the matrix is not unique, but the pairs $(i,j)$ such that $i = \low(j)$ in the 
column echelon form are~\cite{DBLP:books/daglib/0025666}. 
The algorithm requires $O(m^3)$ arithmetic operations in $\Field$.
\subsection{Modular Reconstruction for Elementary Matrix Operations}
Denote by $[r]$ the set of integers $\{1, \cdots ,r\}$. For a
family of $r$ distinct prime numbers $\{q_1, \cdots ,q_r\}$, and a subset
of indices $S \subseteq [r]$, $Q_S$ refers to the product $\prod_{s \in S} q_s$, and we write simply $Q := Q_{[r]}$. For any integer $z \in \Z$ and positive integer $n > 0$, $z \modulo n$ refers to the equivalence class of $z$ in $\Z/n\Z$. For
simplicity, any element $x \in \Z/n\Z$ is identified with the smallest positive integer belonging to the class $x$ in $\Z/n\Z$. We also denote this integer by $x \in \Z$, $0\leq x <m$. Consequently, for $x \in \Z/n\Z$, $x \modulo n'$ refers to the class of $\Z/{n'}\Z$ to which belongs the integer $x \in \Z$, and $(\modulo n')$ can be seen as a ring homomorphism $\Z/n\Z \rightarrow \Z/{n'}\Z$.

\smallskip

We present a particular case of the \emph{Chinese Remainder Theorem}, and recall a simple constructive proof.

\begin{theorem}[Chinese Remainder Theorem~\cite{Gathen:2003:MCA:945759}]
For a family $\{q_1, \cdots ,q_r \}$ of $r$ distinct prime numbers, there exists a ring isomorphism 
\[
\psi: \Z/q_1\Z \times \cdots \times \Z/q_r\Z \rightarrow \Z/{(q_1 \cdots q_r)}\Z
\] 
The isomorphisms $\psi$ and $\psi^{-1}$ can be computed in $O(r)$ arithmetic operations in $\Z/{(q_1 \cdots q_r)}\Z$.
\end{theorem}

\begin{proof}
Euler's theorem states that for two coprime integers $a$ and $n$, $a^{\varphi(n)} \modulo n \allowbreak= 1$, where $\varphi$ is Euler's totient function, which is equal to $\varphi(q)=q-1$ on a prime integer $q$. For $1 \leq s \leq r$, define $\nu_s = (Q_{[r]\setminus\{s\}})^{q_s-1} \modulo Q$. 

For all $1 \leq s \leq r$, there consequently exist integers $\nu_s$ such that $\nu_s \modulo q_t = 1$ if $s=t$ and $0$ otherwise. 
The following expressions of $\psi$ and $\psi^{-1}$ realize the isomorphism of the theorem:
\[
\begin{array}{cccc}
        & \Z/q_1\Z \times \cdots \times \Z/q_r\Z& \leftrightarrow & \Z/(q_1 \cdots q_r)\Z\\
 \psi: &  (u_1, \cdots ,u_r)& \mapsto & \left( u_1 \nu_1+ \cdots +u_r \nu_r \right) \modulo Q\\
 \psi^{-1}:& (x \modulo q_1 , \cdots , x \modulo q_r) &\leftarrow & x \\
\end{array}\] 
\end{proof}

In the following, we consider the isomorphism of the former proof
when referring to the isomorphism given by the Chinese Remainder Theorem.
We denote by $\psi_S$ the function $\psi_S: \prod_{s \in S} \Z/q_s\Z \rightarrow \Z/{Q_S}\Z$ realizing the isomorphism of the Chinese Remainder Theorem for the subset $\{q_s\}_{s\in
S}$, $S \subset [r]$, of prime integers, and we write simply $\psi$ for $\psi_{[r]}$. For a family of elements $u_s \in \Z/q_s\Z, s \in S$, we denote the corresponding $|S|$-tuple $(u_s)_{s\in  S} \in \prod_{s \in S} \Z/q_s\Z$. 
Finally, we recall \emph{Bezout's lemma}~\cite{Gathen:2003:MCA:945759}.

\begin{lemma}[Bezout]
For two integers $a$ and $b$, not both $0$, there exist integers $v$ and $w$ such that $va + wb = \gcd(a,b)$, the greatest common divisor of $a$ and $b$, with $|v|<|b/\gcd(a,b)|$ and $|w|<|a/\gcd(a,b)|$. 

The Bezout's coefficients $(v,w)$ can be computed with the extended Euclidean 
algorithm~\cite{Gathen:2003:MCA:945759}.
\end{lemma}

\medskip

\noindent
{\bf Elementary Column Operations.} We are given a family of distinct prime
numbers $\{q_1,\cdots,q_r\}$, and their product $Q = q_1 \cdots q_r$. Let $\Matrix_Q$ be a matrix with entries in the ring $\Z/Q\Z$. Denoting by $\psi^{-1}: \Z/Q\Z \rightarrow \Z/q_1\Z\times \cdots \times \Z/q_r\Z$ the isomorphism of the Chinese Remainder Theorem, and $\pi_s:\Z/q_1\Z\times \cdots \times \Z/q_r\Z \rightarrow \Z/q_s\Z$ the projection on the $s^{\text{th}}$ coordinate, we call \emph{projection of $\Matrix_Q$ onto $\Z/q_s\Z$}, denoted $\Matrix_Q(\Z/q_s\Z)$, 
the matrix with entries in $\Z/q_s\Z$, obtained by applying $\pi_s \circ \psi^{-1}$ pointwise to each entry of $\Matrix_Q$. 

Conversely, given a number $r$ of $(m \times m)$-matrices $\Matrix_{q_1}, \cdots, \Matrix_{q_r}$ with coefficients in $\Z/q_1\Z , \cdots, \Z/q_r\Z$ respectively, there exists a unique matrix $\Matrix_Q$ with $\Z/Q\Z$ entries such that, for every index $s$ in a prime number $q_s$, $\Matrix_Q$ satisfies $\Matrix_Q(\Z/q_s\Z) = \Matrix_{q_s}$. This is simply a matrix version of the Chinese Remainder Theorem. 

\emph{Elementary column operations} on a matrix $\Matrix$ with entries in a ring $\Ring$ are of three kinds:

\smallskip

\begin{itemize}
\item[(i)] exchange $\col_k$ and $\col_\ell$,
\item[(ii)] multiply $\col_k$ by $-1 \in \Ring$,
\item[(iii)] replace $\col_k$ by $(\col_k + \alpha \times \col_\ell$), for a $\alpha \in \Ring$.
\end{itemize}

\smallskip

For an elementary column operation $(*)$ (i.e., an operation of type (i), (ii) or (iii) applied to some columns of the matrix), we denote by $(*) \circ \Matrix$ the result of applying $(*)$ to $\Matrix$. Any reduction algorithm relies on these three operations. A key feature of the persistent homology reduction is the ability to inverse elements when reducing a matrix with \emph{field} coefficients (applying column operation (iii) in line 4 of the Algorithm~\ref{algo:ph}). 

In the following we introduce algorithms to run elementary column operations simultaneously on matrices $\Matrix_{q_1}, \cdots, \Matrix_{q_r}$ with coefficients in the fields $\Z/q_1\Z , \cdots, \Z/q_r\Z$ respectively, by performing \emph{partial column operations} on matrix $\Matrix_Q$ with coefficient in the ring $\Z/Q\Z$, such that the $\Matrix_{q_j}$ and $\Matrix_Q$ are related by the Chinese Remainder Theorem as above.

Specifically, for an elementary column operation $(*)$, on column $k$, and $\ell$, and with scalar $\alpha \in \Z/Q\Z$, and a subset of indices $S \subseteq [r]$, we call \emph{partial column operation}, denoted by $(*)_S$, on $\Matrix_Q$ the operation transforming $\Matrix_Q$ into $\Matrix'_Q = (*)_S \circ \Matrix_Q$ satisfying: 
\[
\Matrix'_Q \ \text{satisfies} \ 
\left\{ \begin{array} {ll}
\Matrix_Q'(\Z/q_s\Z) = (*) \circ \Matrix_{q_s} & \text{if} \ s \in S,\\
\Matrix_Q'(\Z/q_s\Z) = \Matrix_{q_s} & \text{otherwise}.
\end{array}\right.
\]
where $(*)$ on $\Matrix_{q_s}$ is on column $k$, and $\ell$, and with scalar $\pi_s \circ \psi^{-1}(\alpha)$ in $\Z/q_s\Z$. 
The correspondence $\psi:\Z/q_1\Z \times \cdots
\times \Z/q_r\Z \rightarrow \Z/Q\Z$ is a ring
homomorphism, i.e., it satisfies: 
\[
\psi(u_1,\cdots,u_r)+\psi(v_1,\cdots,v_r)\times \psi(w_1,\cdots,w_r) =
\psi(u_1+v_1\times w_1,\cdots,u_r+v_r\times w_r)
\]
Consequently, we can compute additions and multiplications componentwise in
$\Z/q_1\Z \times \cdots \times \Z/q_r\Z$ using addition and
multiplication in $\Z/Q\Z$. 

\smallskip

In order to compute partial column operations, we first introduce the set of \emph{partial identities}, which are coefficients that allow us to proceed to the partial column operations of type (i) and (ii). 
Secondly, as the rings $\Z/q_s\Z$ are fields, we need to compute the
multiplicative inverse of an element, that is used as multiplicative
coefficient $\alpha$ in elementary column operation (iii). As $\Z/Q\Z$ is not a field, inversion is not possible, and we introduce the concept of \emph{partial inverse} to overcome this difficulty. In the following, the term ``arithmetic operation'' refers to any operation $+$, $-$, $\times$, $\gcd(\cdot,\cdot)$, $(\cdot \mod Q_S)$, and Extended Euclidean algorithm on integer smaller than $Q$. Note they do not have constant time complexity for large $Q$. We discuss arithmetic complexity in Section~\ref{sec:arithcpx}.


\medskip

{\bf Partial Identity and Partial Inverse.} Given a subset of indices $S \subseteq [r]$, we define the \emph{partial identities w.r.t. $S$}, denoted by $L_S$ and equal to
\[
   L_S = \psi( \delta_{1,S}, \cdots , \delta_{r,S} ), \ \text{where} \ 
   \delta_{s,S} \in \Z/{q_s}\Z \ \text{is equal to} 
   \left\{
    \begin{array}{cl}
      1 & \text{if $s \in S$,} \\
      0 & \text{otherwise}.
    \end{array}
   \right.
\]

For any $S \subseteq [r]$, the partial identity $L_S$ can be constructed in $O(r)$ arithmetic operations in $\Z/Q\Z$ by evaluating $\psi$ on $( \delta_{1,S}, \cdots , \delta_{r,S} )$.
However, it is important to notice that if $S = [r]$, $L_{[r]} = \psi(1, \cdots ,1) = 1$, because $\psi$ is 
a ring isomorphism, and $L_{{r}}$ is computed in time $O(1)$.

Knowing the partial identities, we can implement the partial column operations (i) and (ii) for a set of indices $S$. Specifically,

\bigskip

\begin{itemize}
\item[$\text{(i)}$] replace column $\col_k$ by $(\col_k \times L_{[r]\setminus S} + \col_\ell \times L_S)$, and

replace column $\col_\ell$ by $(\col_\ell \times L_{[r]\setminus S} + \col_k \times L_S)$, 
\item[$\text{(ii)}$] multiply column $\col_k$ by $L_{[r]} -2 \times L_S$.
\end{itemize}

\bigskip

As mentioned earlier, we need a notion of ``partial multiplicative inverse'' in $\Z/Q\Z$ in order to pick the appropriate scalar $\alpha$ when defining a partial version of elementary column operation (iii).  We define the \emph{partial inverse} of an element of the ring $\Z/Q\Z$ to be:
\begin{definition}[Partial Inverse]
  Given a set $S \subseteq [r]$ of indices, and an element $x = \psi(u_1, \cdots ,u_r)$ in $\Z/Q\Z$, the \emph{partial inverse of $x$ with regard to $S$} is the element $\overline{x}^S \in \Z/Q\Z$ equal to
  \[
  \overline{x}^S = \psi(\overline{u_1}^S, \cdots , \overline{u_r}^S), \ \
  \text{with} \ \ 
  \overline{u_s}^S =
  \left\{
    \begin{array}{cl}
      u_s^{-1} & \text{if} \ s \in S \ \text{and} \ u_s \in \Z/q_s\Z^\times, \\
      0        & \text{otherwise.}\\
    \end{array}
  \right.
  \]
\end{definition}

We prove elementary arithmetic and computational properties of partial inverses. 
\begin{proposition}[Partial Inverse Construction]
  \label{prop:partialinv}
For a set $S \subseteq [r]$ of indices and an element $x = \psi(u_1, \cdots ,u_r)$ in $\Z/Q\Z$, the following is true:

\bigskip

  \begin{enumerate}[(1)]
  \item $\gcd(x,Q_S) = Q_R$ for some $R \subseteq S$. Additionally, for all $s \in S$, $u_s$ is invertible in $\Z/q_s\Z$ iff $s \notin R$. We denote by $T$ the set $T := S \setminus R$. 
  \item The Bezout's identity for $x$ and $Q_T$ gives $v \cdot x + w \cdot Q_T = 1$, where $v$ satisfies $v \modulo Q_T =  \psi_T((u_s^{-1})_{s \in T})$ 
  \item Finally, 
  \[
  \overline{x}^S = \left[
      \psi_T((u_s^{-1})_{s \in T}) \times L_T \modulo Q
    \right] \in \Z/Q\Z,
  \]
\noindent
 where $L_T$ is the partial identity with regard to $T$. 
  \end{enumerate}
\end{proposition}

\begin{proof}
(1): The $\gcd$ of $x$ and $Q_S$ divides $Q_S$ so $\gcd(x,Q_S) = Q_R$
for some $R \subseteq S$, and for every index $s \in S$, $q_s$
divides $x$ iff $s \in R$. Denote $T := S \setminus R$. According to the Chinese Remainder Theorem, for any
$s \in T$, $u_s = x \modulo q_s \neq 0$ because $q_s$ does not divide $x$. Because $\Z/q_s\Z$ is a field, its unique non invertible element is
$0$ and consequently $u_s$ is invertible. Conversely, because $q_t$ divides
$x$ for $t \in R$, $x \modulo q_t = u_t = 0$ is non invertible.

\medskip

(2): First note that $x \modulo Q_T = \psi_T((u_t)_{t \in T}) \in
\Z/{Q_T}\Z$. Indeed, because $q_t$ divides $Q_T$ for all $t \in T$, we have 
$$\left( x \modulo Q_T \right) \modulo q_t = x \modulo q_t = u_t$$
By definition of $T$, $\gcd(x,Q_T) = 1$ and so the Bezout's lemma gives 
$$v \cdot x + w \cdot Q_T = 1$$
Applying $(\cdot \modulo Q_T)$ to both sides of the
equality gives 
$$(v \modulo Q_T) \psi_T((u_t)_{t \in T}) = 1, \text{and
consequently} \ ((v \modulo Q_T)
\modulo q_t)u_t = 1$$
for every $q_t$ such that $t \in T$. The result follows.

\medskip

(3): Let $L_T$ be the partial identity with regard to $T$. 
We form the product 
$$\widetilde{x} = \left[\psi_T((u_t^{-1})_{t \in T}) \times L_T \modulo Q
\right]$$ 
and evaluate it modulo $q_s$. For any index $s \in [r]$, 
$$\widetilde{x} \modulo q_s = \left[
( \psi_T((u_t^{-1})_{t \in T}) \modulo q_s) \times (L_T \modulo q_s)
\right] \modulo q_s$$
If $s \notin T$, then 
$$L_T \modulo q_s
= 0, \ \ \text{and} \ \ \widetilde{x} \modulo q_s = 0$$
If $s \in T$, then 
$$L_T \modulo q_s = 1, \ \ \text{and} \ \ \psi_T((u_t^{-1})_{t \in T}) \modulo q_s = u_s^{-1}$$ 
and
consequently $\widetilde{x} \modulo q_s =  u_s^{-1}$. Thus, $\widetilde{x}$
satisfies the definition of $\overline{x}^S$, the
partial inverse of $x$ with regard to $S$.  
\end{proof}

We directly deduce an algorithm to compute the partial inverse of $x$
w.r.t $S$ if $Q_S$ is given: compute
$Q_R = \gcd(x,Q_S)$ and $Q_T=Q_S/Q_R$, then $v$ using the extended
Euclidean algorithm and finally $\overline{x}^S = (v \modulo Q_T)\times
L_T \modulo Q$. Computing the partial identity $L_T$ requires $O(r)$ arithmetic operations in 
$\Z/Q\Z$, but is constant if $T = [r]$, which happens iff $S = [r]$ and $x$ is invertible in $\Z/Q\Z$. 
Consequently, computing $\overline{x}^S$ requires $O(r)$ arithmetic operations in general, but 
only $O(1)$ arithmetic operations in the latter case.

\subsection{Modular Reconstruction for Multi-Field Persistent Homology}
\label{subsec:mf_modrec_phom}

Let $\K$ be a filtered complex of size $m$. Define $\Matrix_\partial(\Z/q_s\Z)$ to be the $(m \times m)$ boundary matrix of $\K$ with $\Z/q_s\Z$ coefficients. Define $\Matrix$ to be the $(m \times m)$ matrix with $\Z/Q\Z$ coefficients such that the projection of $\Matrix$ onto $\Z/q_s\Z$ is equal to $\Matrix_\partial(\Z/q_s\Z)$, for all $s \in [r]$. Note that the matrices $\Matrix$ and $\Matrix_{\partial}(\Z/q_s\Z)$, for any $s$, are ``identical'' matrices in the sense that they contain $0$, $1$ and $-1$ coefficients at the same positions, where $0$, $1$ and $-1$ refer respectively to elements of $\Z/Q\Z$ and $\Z/q_s\Z$. 

We reduce a matrix $\Rmat$ which is initially equal to $\Matrix$. Denote by $\col_j$ the $j^{\text{th}}$ column of $\Rmat$. Define the \emph{extended $\low$ function $\low(j,Q_S)$} to be the index of the 
lowest element of $\col_j$ such that $\col_j[\low(j,Q_S)] \mod Q_S \neq 0$. In particular, $\low(j,q_s)$ is equal to the index of the lowest non-zero element of column $j$ in the projection $\Rmat(\Z/q_s\Z)$, and $\low(j,Q_S)$ is equal to 
\[
\low(j,Q_S) = \max_{s \in S} \low(j,q_s)
\]

After iteration $j$, we say that the columns $\col_1, \cdots, \col_j$ are \emph{reduced}. We maintain, for every reduced column $\col_j$, the collection of ``lowest indices'' $i$ as a set $\mathcal{L}(j) = \{ (i,Q_S) \}$ satisfying three conditions ensuring that $\low$ values for all indices $s \in [r]$ are represented, without redundancy. Specifically, the set $\mathcal{L}(j)$ satisfies: 

\smallskip 

\begin{itemize}
\item[-] For every $(i,Q_S) \in \mathcal{L}(j)$, $i = \low(j)$ in matrix $\R(\Z/q_s\Z)$ for every $s \in S$.
\item[-] Every two distinct pairs $(i,Q_S), (i', Q_{S'}) \in \mathcal{L}(j)$ satisfy both $i \neq i'$ and $S \cap S' = \emptyset$.
\item[-] The union $\displaystyle\cup_{(i,Q_S) \in \mathcal{L}(j)} S = [r]$. 
\end{itemize}

\smallskip

The algorithm is presented in Algorithm~\ref{algo:modrec}. It returns the set of triplets $\mathcal{P}_r = \{(i,j,Q_S)\}$ such that $i = \low(j)$ in the column echelon form of the matrix $\Matrix_{\partial}(\Z/q_s\Z)$ iff $s \in S$, or, equivalently, $(i,Q_S) \in \mathcal{L}(j)$ 
once $\col_j$ has been reduced. This is a compact encoding of the \emph{multi-field persistence diagram}. Note that it contains exactly $P_r$ elements. 

\begin{algorithm}[t]
\SetAlgoLined
\KwData{Matrix $\Rmat = \Matrix$, diagram $\mathcal{P}_r \leftarrow \emptyset$}
\KwOut{Multi-field persistence diagram $\mathcal{P}_r = \{(i,j,Q_S)\}$}
\For{$j=1, \cdots ,m$}
{
  $Q_S \leftarrow Q_{[r]}$\;                                             
  \While{$\low(j,Q_S)$ is defined}
  {
    $k \leftarrow \low(j,Q_S)$;  $\ \ \ \ Q_T \leftarrow Q_S / \gcd(\col_j[k], Q_S)$ \;
    \While{there exists $j' < j$ with $(i,Q_{T'}) \in \mathcal{L}(j')$\\
    \hfill satisfying $\left[ i = \low(j,Q_S) \text{ and } \gcd(Q_{T'},Q_T)>1 \right]$}
    {
		$Q_T \leftarrow Q_T / \gcd(Q_{T'},Q_T)$ \;
		$\col_j \leftarrow \col_j - \left(\col_j[k] \cdot \overline{\col_{j'}[k]}^T \right) \times \col_{j'}$ \;
    }
    \lIf{$Q_T \neq 1$}{$\mathcal{P}_r \leftarrow \mathcal{P}_r \cup \{(k,j,Q_T)\}$;
                  $\ \ \ $ $Q_S \leftarrow Q_S/Q_T$ \;}
  }
}
\caption{Simultaneous persistent homology algorithm for $\Z/q_1\Z, \ldots \Z/q_r\Z$}
\label{algo:modrec}
\end{algorithm}

The $\{\mathcal{L}(j)\}_j$ form an index table that we maintain implicitly.
At iteration $j$ of the {\tt for} loop, we use $Q_S$ for the product of all prime numbers 
$\prod_{s\in S} q_s$ for which the column $j$ in $\Rmat(\Z/q_s\Z)$ has not yet been reduced. 

\medskip

\noindent
{\bf Analysis.} We give details on the line-by-line computation of Algorithm~\ref{algo:modrec} in terms of operations induced in the matrices $\Rmat(\Z/q_s\Z)$ for $s \in [r]$. A set of indices $S \subset [r]$ is maintained by storing the product $Q_S$, and set operations, such as set difference and set intersection, are implemented using respectively arithmetic division and greatest common divisor. Specifically, for $T \subset S \subset [r]$, and $T'\subset [r]$,
\[
    Q_S / Q_T = Q_{S \setminus T} \ \ \ \text{and} \ \ \ \gcd(Q_{T}, Q_{T'}) = Q_{T \cap T'}
\]

The set $S$ in the {\bf while} loop line $3$ contains exactly the set of indices $s \in S$ such that the column $\col_j$ of matrix $\Rmat(\Z/q_s\Z)$ is not yet reduced. In line $4$, $k$ is the lowest row index of a $\col_j$ in any of the matrices $\Rmat(\Z/{q_t}\Z)$ such that $t \in S$ (i.e., a matrix in which $col_j$ is still unreduced). The matrices $\Rmat(\Z/q_t\Z)$ where $\low(j)$ is exactly $k$ are the ones for which $t \in T$ (line $4$). This property of $T$ is maintained over all of the {\bf while} loop line $3$.
 
The set $T'$ defined on line $5$ contains some indices $t$ such that $\col_j$ and $\col_{j'}$ have same lower index $i$ in $\Rmat(\Z/q_t\Z)$. 
By definition of the partial inverse and the sets $T$ and $T'$, the column operation line $8$ modifies only the matrix $\Rmat(\Z/q_t\Z)$ for $t \in T \cap T'$ and reduces strictly their $\low(j)$ values. In line $7$, the set $T$ is updated to contain exactly the indices $t$ such that $\low(j) = k$ in $\Rmat(\Z/q_t\Z)$. 

At line $10$, all columns $\col_j$ in $\Rmat(\Z/q_t\Z), t \in T$, are reduced and non-zero, we update the multi-field persistence diagram and maintain the property that $S$ contains exactly the indices $s$ for which $\col_j$ is still unreduced in $\Rmat(\Z/q_s\Z)$.

\medskip

\noindent
{\bf Correctness.} First, note that all operations processed on $\Rmat$ correspond to left-to-right elementary column operations in the matrices $\Rmat(\Z/q_s\Z)$ for all $s \in [r]$. 
One iteration of the {\tt while} loop in line $3$ either strictly reduces $Q_S$ by dividing it by $Q_T$ (when $T \neq \emptyset$ in line $10$) or sets $(\col_j[k] \mod Q_S)$ to zero thus reducing strictly $\low(j,Q_S)$ (when $T = \emptyset$ and $Q_T = 1$). Consequently, the algorithm terminates.

We prove recursively, on the number of columns, that each of the matrices $\Rmat(\Z/q_s\Z)$ gets reduced to column echelon form. We fix an arbitrary field $\Z/q_s\Z$: suppose that the $j-1$ first columns of $\Rmat(\Z/q_s\Z)$ have been reduced at the end of iteration $j-1$ of the {\tt for} loop in line $1$. We prove that at the end of the $j^{\text{th}}$ iteration of the {\tt for} loop in line $1$, the $j$ first columns of the matrix $\Rmat(\Z/q_s\Z)$ are reduced. Consider two cases. 

1. First suppose that there is a triplet $(i,j,Q_T)$ in the multi-field persistence diagram $\mathcal{P}_r$, for some $i < j$ and $Q_T$ satisfying $q_s$ divides $Q_T$. 
This implies that the algorithm exits the {\tt while} loop line $5$ with $q_s$ dividing $Q_T$, and $Q_T$ dividing $Q_S$ (because by definition of $Q_T$, in line $4$, $Q_T$ divides $Q_S$) and there is no $j' < j$ such that $\left[\low(j',Q_{T'}) = \low(j,Q_S)\right.$ and $\left[\gcd(Q_{T'},Q_T)>1 \right]$. This in particular implies that there is no $j' < j$ such that $\low(j',q_s) = \low(j,q_s)$ and column $j$ is reduced in $\Rmat(\Z/q_s\Z)$.

2. Secondly, suppose that there is no such pair $(i,j,Q_T)$ in $\mathcal{P}_r$, with $q_s$ dividing $Q_T$. Consequently, during all the computation of the {\tt while} loop in line $3$, $q_s$ divides $Q_S$. When exiting this {\tt while} loop, $\low(j,Q_S)$ is undefined, implying in particular that $\low(j,q_s)$ is 
undefined and column $j$ of $\Rmat(\Z/q_s\Z)$ is zero, and hence reduced.

\medskip

\noindent
{\bf Reconstruction of Cycles and Pairs.} Denote by $\Rmat(\Z/q_s\Z)$ the matrix maintained at iteration $i$ of the standard persistent homology Algorithm~\ref{algo:ph}
with coefficient in the field $\Z/q_s\Z$. Note that, at iteration $i$ of the modular reconstruction Algorithm~\ref{algo:modrec}, we maintain a matrix $\Rmat$ that is a compact representation of all matrices $\Rmat(\Z/q_s\Z)$, for $s = 1 \ldots r$. Indeed, applying $(\cdot \mod q_s)$ to all coefficients of $\Rmat$ leads to a matrix $\Rmat(C)$. Consequently, we can reconstruct the cycles and the persistent pairs for standard persistent homology with $\Z/q_s\Z$ coefficients, for any $s = 1 \ldots r$, with the modular reconstruction algorithm.

\section{Output-Sensitive Complexity Analysis} 
\label{sec:mf_cpx}

We start by describing a complexity model for the arithmetic operation on large integers.


\subsection{Arithmetic Complexity Model for Large Integers}
\label{sec:arithcpx}

During the reduction algorithm we perform arithmetic operations on big integers, for which we describe a complexity model~\cite{Gathen:2003:MCA:945759}. Suppose that on our architecture, a memory word is encoded on $\wordlength$ bits (on modern architectures, $\wordlength$ is usually $64$). Computer chips contain Arithmetic Logic Units that allow arithmetic operations on a $1$-memory word integer in $O(1)$ machine cycles. Let the \emph{length} of an integer $n$ be defined by: $\lambda(n) = \left\lfloor \log_2 n/\wordlength \right\rfloor +1$, i.e., by the number of memory words necessary to encode $n$. We express the arithmetic complexity as a function of the length. For any positive integer $n$ of length $\lambda(n)=\bits$, operations in $\Z/n\Z$ cost $\Acpx_+(n)=O(\bits)$ for additions, $\Acpx_\times(n)=O(M(\bits))$ for multiplications, and $\Acpx_\div(n)=O(M(\bits)\log \bits)$ for the (extended) Euclidean algorithm, inversions and divisions, where $M(\bits)$ is a monotonic upper bound on the number of word operations necessary to multiply two integers of length $\bits$~\cite{Gathen:2003:MCA:945759}. The best known upper bound~\cite{DBLP:journals/siamcomp/Furer09} is $M(\bits) = O(\bits \log \bits \ 2^{O(\log ^* \bits)})$, where $\log ^* \bits$ is the iterated logarithm of $\bits$. 

In the case of multi-field persistent homology, we are interested in the value of $\lambda$ for an element in $\Z/Q\Z$, $Q=q_1 \cdots q_r$, in the case where $\{q_1, \cdots ,q_r\}$ are the first $r$ prime numbers. By virtue of the inequalities~\cite{RS62} $\ln Q < 1.01624 q_r$, and $q_r < r \ln(r\ln r)$ for $r\geq 6$, the number of bits to encode a scalar in $\Z/Q\Z$ (as an integer between $0$ and $Q-1$) is $\lambda(Q) < \lfloor 1.46613 \ r \ln(r \ln r) / \wordlength \rfloor +1$. 

\medskip

\noindent
{\bf Notation.} We denote by $\Acpx_r$ an upper bound on the time complexities $\Acpx_+(Q)$, $\Acpx_\times(Q)$, and $\Acpx_\div(Q)$, for performing arithmetic operations on integers smaller than $Q$, where $Q = q_1 \times \ldots \times q_r$ is the product of the $r$ smallest prime numbers $q_1, \ldots , q_r$.

\subsection{Complexity of the Modular Reconstruction Algorithm}
\label{sec:cpx_modrec}

Let $\K$ be a filtered complex of size $m$. We describe computational complexities in terms of the size of the persistence diagram $P_\Field$, which is a $\Theta(m)$, and the size of the multi-field persistence diagram $P_r$. The persistent homology algorithm described in Section~\ref{sec:mf_algo}, applied on $\K$ with coefficients in a field $\Field$, requires $O(P_\Field^3)$ operations in $\Field$. For a field $\Z/q\Z$ these operations take constant time and the algorithm has complexity $O(P_\Field^3)$. The output of the algorithm is the persistence diagram.

For a set of prime numbers $\{q_1, \cdots ,q_r\}$, let $P_r$ be the total number of distinct pairs in all persistence diagrams for the persistent homology of $\K$ with coefficient fields $\Z/q_1\Z, \cdots, \Z/q_r\Z$. We express the complexity of the modular reconstruction algorithm in terms of the size of its output $P_r$, the number of fields $r$ and the arithmetic complexity $\Acpx_r$. 

First, note that, for a column $j'$ in the reduced form of $\Rmat$, the size of $\mathcal{L}(j')$ is equal to the number of triplets of the multi-field persistence diagram with death index $j'$. We denote this quantity by $|\mathcal{L}(j')|$. Hence, when reducing column $\col_j$ with $j > j'$, the column $\col_{j'}$ is involved in a column operation $\col_j \leftarrow \col_j + \alpha \cdot \col_{j'}$ at most $|\mathcal{L}(j')|$ times. Consequently, reducing $\col_j$ requires $O(\sum_{j' < j} |\mathcal{L}(j')|) = O(P_r)$ column operations. There is a total number of $O(m \times P_r)$ column operations to reduce the matrix, each of them being computed in time $O(m \times \Acpx_r)$.

Computing the partial inverse of an element $x \in \Z/Q\Z$ takes time $O(r \times \Acpx_r)$ in the general case, and only $O(\Acpx_r)$ if $x$ is invertible in $\Z/Q\Z$. The partial inverse of an element $x = \col_j[k]$ is computed only if there is a pair $(k,Q_T) \in \mathcal{L}(j)$. This element is not invertible in $\Z/Q\Z$ iff $|\mathcal{L}(j)| > 1$. There are consequently $O(|P_r - m|)$ non-invertible elements $x$ that are at index $\low(j,Q_T)$ in some column $j$, for some $Q_T$. If we store the partial inverses when we compute them, the total complexity for computing all partial inverses in the modular reconstruction algorithm is $O(m + r \times (P_r - m) \times \Acpx_r)$.
We conclude that the total cost of the modular reconstruction algorithm for multi-field persistent homology is 

\[
O\left( \left[ r \times (P_r - m) + m^2 P_r \right] \times \Acpx_r) = O( \left[ r \times (P_r-P_\Field) + P_r^3 \right] \times \Acpx_r\right)
\]

\noindent
while the brute-force algorithm, consisting in computing persistence separately for every field $\Z/q_1\Z, \cdots ,\Z/q_r\Z$ has time complexity 
\[O(r \times P_\Field^3)\]

\subsection{Discussion and Limitations}
\label{sec:disc_cpx}

Comparing the time complexity of the modular reconstruction algorithm and the brute-force approach, we notice that the former in particularly more efficient than the latter when $\Acpx_r$ is not too large, and the difference between persistence diagrams for different coefficient fields are few. In that case, assuming $r \times (P_r-P_\Field) \ll P_r^3$, the trade-off of time complexities is
\[
  \frac{r}{\Acpx_r}
\]

In light of Section~\ref{sec:arithcpx}, the complexity $\Acpx_r$ in practice is a near-linear function in the number of memory word $\lambda(Q)$ necessary to store the integer $Q = q_1 \cdots q_r$, for the first $r$ prime numbers. In particular, we note that $\lambda(Q)\ll r$ for $r \ln r \ll e^\wordlength$.

We note two limitations to the modular reconstruction algorithm. First, for large numbers $r$ of primes, one arithmetic operation in $\Z/{Q_{[r]}}\Z$ becomes more costly than $r$ distinct arithmetic operations in $\Z/q_1\Z, \cdots, \Z/q_r\Z$, in which case the modular reconstruction approach developed in this article becomes worse than brute-force (even when $P_r$ and $P_\Field$ remain close). 

Second, for complexes with torsion subgroups of very high order in their homology, the number of distinct pairs $P_r$ in all the persistence diagram may become large.

We study these cases in practice in the experimental Section~\ref{sec:mf_expe}.

\section{Complexity Analysis in Terms of Index Persistence}
\label{app:mf_barcomplexity}

The cubic dependence in the size of the persistence diagram is, in practice, pessimistic. In this section we refine the complexity analysis in terms of the length of persistence intervals, in the spirit of the sparse complexity analysis of the standard persistence algorithm~\cite{DBLP:books/daglib/0025666}. First, we recall the sparse complexity analysis of the persistent homology algorithm.

\begin{theorem}[Sparse Complexity Analysis PH~\cite{DBLP:books/daglib/0025666}[Chapter VII] ] 
With a sparse matrix implementation, where only non-zero matrix coefficients are represented, the algorithm reduces the boundary matrix of a filtered simplicial complex of dimension $d$ in 
\[
O\left(d \times \left[ \sum_{ (i,j) \in \mathcal{P}_\Field, j \neq \infty } |j-i|^2 + \sum_{(i, \infty) \in \mathcal{P}_\Field} i^2 \right] \right)
\]

\noindent 
arithmetic operations in $\Field$.
\label{thm:mf_indexpersistence_cpx}
\end{theorem}

\begin{proof}
The proof is identical to the one in~\cite{DBLP:books/daglib/0025666}, except that the ``clear'' optimization of~\cite{Chen11persistenthomology} (see also~\cite{Bauer:arXiv1303.0477}) allows us to improve the bound. The argument of the proof relies on the fact that: 

1. to reduce a column $\col_j$, eventually leading to an interval $(i,j)$ in the diagram, only columns $\col_{j'}$ with $i < j' < j$ are used for the reduction.

2. to reduce a column $\col_i$, eventually leading to an interval $(i, \infty)$ in the diagram, only columns $\col_{j'}$ with $j' < i$ are used for the reduction.

3. any column $\col_i$ such that $i$ is the birth index of a finite interval in the diagram can be reduced in $O(1)$ operations using the clear optimization.

The complexity bound can be read directly from this analysis.
\end{proof}

We can deduce almost directly from the proof of Theorem~\ref{thm:mf_indexpersistence_cpx} the following:

\begin{corollary}
With a sparse matrix implementation, where only non-zero matrix coefficients are represented, the modular reconstruction algorithm for multi-field persistent homology, with coefficient fields $\Z/q_1\Z, \cdots, \Z/q_r\Z$, applied on a filtered simplicial complex of dimension $d$, has complexity:
\[
O \left( 
\begin{array}{l} 
   \Acpx_r \times r \times (P_r-P_\Field) \ + \\
   \Acpx_r \times d \times \left[ \sum_{ (i,j,Q_S) \in \mathcal{P}_r, j \neq \infty } |j-i|^2 \times |\mathcal{L}(j)| + \sum_{(i, \infty, Q_S) \in \mathcal{P}_r} i^2 \times |\mathcal{L}(i)| \right]
\end{array}  \right)
\]

\noindent
where $|\mathcal{L}(j)|$ is the number of triplets $(i,j,Q_S)$ of the multi-field persistence diagram $\mathcal{P}_r$ dying at index $j$.
\end{corollary}

\begin{proof}
We note that, when reducing a column $\col_j$ in the modular reconstruction algorithm, a column $\col_{j'}$ is added at most $|\mathcal{L}(j)|$ times to $\col_j$, with different multiplicative weights. The rest of the proof is identical to the proof of Theorem~\ref{thm:mf_indexpersistence_cpx}.
\end{proof}

\section{Experiments and Applications}
\label{sec:mf_expe}

In this section we report the performance of the modular reconstruction algorithm for multi-field persistent homology against the brute-force approach consisting in computing persistent homology separately for every field of coefficients. Our implementation is in {\tt C++} and is available within the {\tt Gudhi} software library~\cite{gudhi:PersistentCohomology} for topological data analysis. We use the {\tt GMP} library~\cite{gmplib_cite} for storing large integers. All timings are measured on a $64$ bits Linux machine with $3.00$ GHz processor and $32$ GB RAM., and are averaged over $10$ independent runs. 

We compute the persistent homology of Rips complexes~\cite{DBLP:books/daglib/0025666}, which are one of the most popular constructions in topological data analysis, built on a variety of both real and synthetic geometric data, and we compute the persistent homology of a variety of random simplicial complexes. We use the \emph{compressed annotation matrix} implementation of persistence~\cite{DBLP:journals/algorithmica/BoissonnatDM15} for its efficiency and stability over various datasets. Additionally, the compressed annotation matrix is one of the fast implementations of persistent homology that use few arithmetic operations.

\subsection{Description of the Data}

Datasets and running times are presented in Figure~\ref{fig:mf_timings}.

\medskip

\noindent
{\bf Topological Data Analysis.} We use a variety of natural and synthetic geometric data for the running times: {\bf Bud} is a set of points sampled from the surface of the {\it Stanford Buddha} in $\R^3$. {\bf Bro} is a set of $5\times 5$ {\it high-contrast patches} derived from natural images, interpreted as vectors in $\R^{25}$, from the Brown database (with parameter $k=300$ and cut $30\%$)~\cite{DBLP:journals/ijcv/CarlssonISZ08}. {\bf Cy8} is a set of points in $\R^{24}$, sampled from the space of conformations of the cyclo-octane molecule~\cite{martin2010top}, which is the union of two intersecting surfaces. {\bf Kl} is a set of points sampled from the surface of the figure eight Klein Bottle embedded in $\R^5$. Finally {\bf S3} is a set of points distributed on the unit $3$-sphere in $\R^4$. Datasets are listed in Figure~\ref{fig:mf_timings} with the size of point sets $|P|$, the ambient dimensions $D$ and intrinsic dimensions $d$ of the sample points (if known), the thresholds $\rho$ for the Rips complex and the size of the complexes constructed $|\mathcal{K}|$. 

In topological data analysis, data points are generally geometric samples of low-dimensional spaces---such as manifolds---embedded in high-dimensions. Their persistence usually show few (or none) long living torsion, of low order. 

\medskip

\noindent
{\bf Random Complexes.} We use three distinct models of random simplicial complexes ; see for example~\cite{Bobrowski2018} for a survey on random complexes. The complex $\mathbf{\mathcal{R}(10000,0.25)}$ is the 5-skeleton of a Rips complex on $10000$ uniform random points in the unit cube in $\R^5$, with threshold $0.25$, where the filtration in the standard (geometric) Rips filtration. $\mathbf{X(200,5000)}$ is the 5-skeleton of a random flag complex on $200$ vertices with $5000$ random edges, where the filtration is induced by an ordering of the edges. $\mathbf{Y_2(50,3000)}$ is a Linial-Meshulam random 2-complex on $50$ vertices with $3000$ random triangles, where the complex is filtered by an ordering of the triangles. 

These complexes usually show a lot of torsion in their persistence, that may be of high order. In particular, Linial-Meshulam random 2-complexes on $n$ points are known to show experimentally a burst of torsion, of potentially super-exponential order in $n$.


%
\begin{figure}[t]
\centering
\setlength{\tabcolsep}{2.8pt}
\begin{tabular}{|l | r | r | r | r | r | r | r | r | r | r | r | r | r|}
\hline
Data & $|P|$ & $D$ & $d$ & $\rho$ & $|\mathcal{K}|$ & $T_1$ & $\text{R}_1$ & $T_{50}$ & $\text{R}_{50}$ & $T_{100}$ &
$\text{R}_{100}$ & $T_{200}$ & $\text{R}_{200}$ \\
\hline
{\bf Bud} & 49,990 & 3  & 2 & 0.09 & $127 \cdot 10^6$ & 96.3  & 0.51 & 110.3 & 22.2 & 115.9 & 42.3 & 130.7 & 75.0\\
{\bf Bro} & 15,000 & 25 & ? & 0.04 & $142 \cdot 10^6$ & 123.8 & 0.41 & 143.5 & 17.8 & 150.2 & 34.0 & 174.5 & 58.5\\
{\bf Cy8} & 6,040  & 24 & 2 & 0.8  & $193 \cdot 10^6$ & 121.2 & 0.63 & 134.6 & 28.2 & 139.2 & 54.6 & 148.8 & 102.2\\
{\bf Kl}  & 90,000 & 5  & 2 & 0.25 & $114 \cdot 10^6$ & 78.6  & 0.52 & 89.3  & 23.0 & 93.0  & 44.1 & 105.2 & 78.0\\
{\bf S3}  & 50,000 & 4  & 3 & 0.65 & $134 \cdot 10^6$ & 125.9 & 0.40 & 145.7 & 17.2 & 152.6 & 32.8 & 177.6 & 50.3 \\
\hline
\end{tabular}
\setlength{\tabcolsep}{6.3pt}
\begin{tabular}{|l | r | r | r | r | r | r | r | r | r |}
\hline
Data & $D$ & $T_1$ & $\text{R}_1$ & $T_{50}$ & $\text{R}_{50}$ & $T_{100}$ & $\text{R}_{100}$ & $T_{200}$ & $\text{R}_{200}$ \\
\hline
$\mathbf{\mathcal{R}(10000,0.25)}$ & 5 & 11.6 & 0.49 & 14.9 & 20.0 & 15.7 & 37.1 & 19.4 & 61.1 \\
$\mathbf{X(200,5000)}$ & 5 & 10.55 & 0.52 & 34.9 & 21.6 & 47.5 & 29.0 & 67.6 & 41.1 \\
$\mathbf{Y_2(50,3000)}$ & 2 & 0.22 & 0.42 & 1.55 & 4.69 & 3.36 & 4.8 & 6.9 & 3.8 \\
\hline
\end{tabular}
\caption{Timings $T_r$ of the modular reconstruction algorithm for the first $r$ prime numbers, and ratio $R_r$ with the brute-force algorithm. Top: Rips complexes on geometric data from topological data analysis. Bottom: Diverse models of random complexes.}
\label{fig:mf_timings}
\end{figure}

\subsection{Time Performance of the Algorithm}

In Figure~\ref{fig:mf_timings}, the values $T_r$ for $r \in \{1, 50, 100, 200\}$ refers to the running time of the modular reconstruction algorithm for the $r$ first prime numbers, and $\text{R}_r$ refers to the ratio between the timings of the brute-force approach (cumulating timings for persistence in every coefficient field), and the timings of the modular reconstruction algorithm. Timings are average over $10$ independent running times, picking up new instances of complexes for the random complexes.

\medskip

\noindent
{\bf Topological Data Analysis.} Interestingly, we observe that on all experiments the number of differences between persistence diagrams with various coefficient fields is small. Following Section~\ref{sec:disc_cpx}, the quantity $P_r-P_\Field$ can be considered to be a small constant in our experiments. We have also observed that these differences appeared for small prime numbers $q_s$. 
Consequently, the linear dependence in $r$ from component $r \times (P_r - P_\Field)$ of the complexity analysis in Section~\ref{sec:mf_cpx} is negligible experimentally. We can consider that, experimentally, the ratio between the brute-force timings and the modular reconstruction timings is at most 
\[
    \frac{r}{\Acpx_r}
\]
where, in light of the discussion of Section~\ref{sec:arithcpx}, $\Acpx_r$ is a small constant for small to medium values of $r$ (here, $r \leq 200$). Specifically, for $q_1, \ldots , q_r$ the $r$ first prime numbers and on a $64$ bits machine, the number of memory words necessary to represent the product $Q = q_1 \times \ldots \times q_r$ is $\lambda(Q) = 7, 15, \text{ and } 32$ for $r = 50, 100, \text{ and } 200$ respectively. Additionally, the optimized implementation of persistent homology using fewer arithmetic operations, the trade-off $r / \Acpx_r$ is pessimistic. These considerations are confirmed by the experiments.

\smallskip

Figure~\ref{fig:mf_timings} presents the timings of the modular reconstruction approach for a variety of filtered simplicial complexes ranging between $114$ and $193$ million simplices. We note that from $r=1$ to $r=200$ prime numbers, the time for computing multi-field persistence using the modular reconstruction approach only increases by $23$ to $41\%$, when the brute-force approach requires about $200$ times more time, as expected. This difference appears in the speed-up expressed by the ratio $R_r$. 
For $r=1$, the modular reconstruction approach is about twice slower than the standard persistent homology algorithm in one field, because modular reconstruction is a more complex procedure and deals, in our implementation, 
with {\tt GMP} integers that are slower than the classic {\tt int} used in the standard persistent homology algorithm. However, this difference fades away as soon as $r>1$ and the modular reconstruction is significantly more efficient than brute-force: it is, in particular, between $50.3$ and $102.2$ times faster for $r=200$. We study the asymptotic behaviour of the running times for large values of $r$ in Section~\ref{subsec:asymptotics}.

\medskip

\noindent
{\bf Random Complexes.} Figure~\ref{fig:mf_timings} presents timings for the modular reconstruction on random complexes. A similar analysis as the one for geometric data holds for the random complexes $\mathbf{\mathcal{R}(10000,0.25)}$ and $\mathbf{X(200,5000)}$, despite the appearance of more torsion in their persistent homology. Indeed, for $r=1$ we observe that the modular reconstruction algorithm is about twice slower due to the manipulation of {\tt GMP} integers, but for increasing values of $r$ the modular reconstruction approach gets faster, and is in particular between $41$ and $61$ times faster for $r=200$.

The case of the random Linial-Meshulam complex $\mathbf{Y_2(50,3000)}$ shows the limit of the approach, and the difference of running times is not as remarkable. These complexes show short torsion in their persistent homology (see Section~\ref{subsec:app} for an analysis) but the torsion subgroups of $H_1(Y_2)$ are of very high order. Following Section~\ref{sec:disc_cpx}, the difference $P_r - P_\Field$ in the complexity analysis increases for larger values of $r$, becoming non-negligible and hence slowing down the modular reconstruction algorithm.

\begin{figure}[!t]
  \centering
      \begin{tikzpicture}[scale=0.75] 
      \begin{axis}[
        scaled ticks=false,
        legend style={ at={(0.015,0.9)}, anchor=west},
        grid=major,
        width=13cm,
        height=9cm,
        xmin=0,
        xmax=151,
        ymin=-7,
        ymax=330,
        xlabel= number of primes $r$,
        ylabel = time (s.)]        
        \addplot
        coordinates {
            (1,   6.46 )
            (10,  6.52 )
            (20,  7.00 )
            (30,  7.52 )
            (40,  8.01 )
            (50,  8.07 )
            (60,  7.76 )
            (70,  8.00 )
            (80,  8.10 )
            (90,  8.26 )
            (100, 8.31 )
            (110, 8.41 )
            (120, 9.08 )
            (130, 9.54 )
            (140, 9.44 )
            (150, 9.56 )
        };   \addlegendentry{modular reconstruction}
        \addplot[red,mark = square*] 
        coordinates {   
            (1,  2.13  )
            (10, 21.3  )
            (20, 42.6  )
            (30, 63.9  )
            (40, 85.2  )
            (50, 106.5 )
            (60, 127.8 )
            (70, 149.1 )
            (80, 170.4 )
            (90, 191.7 )
            (100, 213  )
            (110, 234.3)
            (120, 255.6)
            (130, 276.9)
            (140, 298.2)
            (150, 319.5)
        };  \addlegendentry{brute force}
      \end{axis}

      \begin{axis}[
        scaled ticks=false,
        legend style={ at={(0.015,0.7)}, anchor=west},
        width=13cm,
        height=9cm,
        xmin=0,
        xmax=101,
        hide x axis,
        ymin=0,
        ymax=40,
        axis y line=right, 
        ylabel = ratio]
        \addplot[black,mark=triangle*] 
        coordinates { 
            (1   ,0.33)
            (10  ,3.27)
            (20  ,6.09)
            (30  ,8.50)
            (40  ,10.64)
            (50  ,13.20)
            (60  ,16.47)
            (70  ,18.64)
            (80  ,21)
            (90  ,23.21)
            (100 ,25.63)
            (110 ,27.86)
            (120 ,28.15)
            (130 ,29.03)
            (140 ,31.59)
            (150 ,33.42)

        };\addlegendentry{ratio}
      \end{axis}
      \end{tikzpicture} 

      \caption{Timings for the modular reconstruction
        algorithm and brute force.}
      \label{fig:mf_plotsmall}
\end{figure}
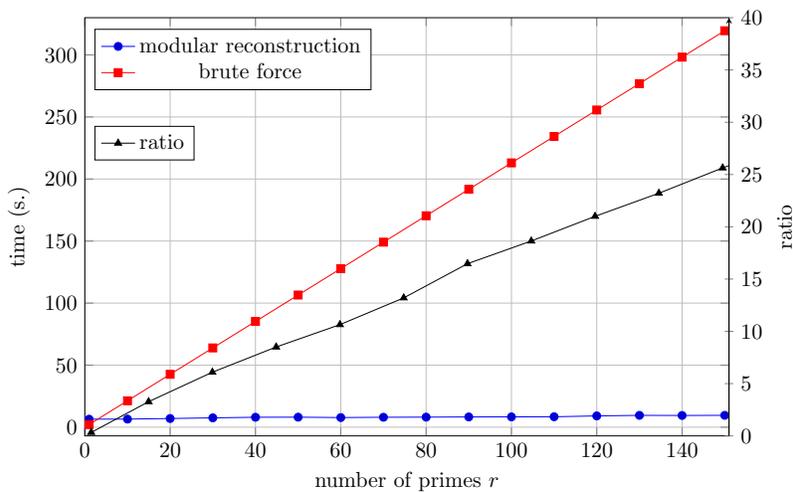

\begin{figure}
    \centering

    \begin{tikzpicture}[scale=0.75] 
    \begin{axis}[
            legend style={ at={(0.015,0.94)}, anchor=west},
            scaled ticks=false,
            grid  =major,
            width =13cm, 
            height=9cm,
            xmin=0,
            xmax=100000,
            xtick = {0,30000,60000,90000},
            ymin=0,
            ymax=4000,
            xlabel= number of primes $r$,
            ylabel = time (s.)]     ] 
        \addplot coordinates {
        (1      , 6.46)
        (10000  , 235.1)   
        (20000  , 540.2)
        (30000  , 896.6)
        (40000  , 1253.5)
        (50000  , 1649.6) 
        (60000  , 2047.2)
        (70000  , 2467.0)
        (80000  , 2898.3)
        (90000  , 3342.6)
        (100000 , 3840.0)    }; \addlegendentry{modular reconstruction}
    \end{axis}

      \begin{axis}[
            scaled ticks=false,
            legend style={ at={(0.015,0.81)}, anchor=west},
            width=13cm,
            height=9cm,
            xmin=0,
            xmax=100000,
            hide x axis,
            ymin=0,
            ymax=130,
            axis y line=right, 
            ylabel = ratio]
            \addplot[black,mark=triangle*] 
          coordinates {
                        (10000  , 90.1)   
                        (20000  , 78.9)
                        (30000  , 71.3)
                        (40000  , 68.0)
                        (50000  , 64.5) 
                        (60000  , 62.4)
                        (70000  , 60.4)
                        (80000  , 58.8)
                        (90000  , 57.4)
                        (100000 , 55.5)    }; \addlegendentry{ratio}

    \end{axis}
    \end{tikzpicture} 

\caption{Asymptotic behaviour of modular reconstruction and brute force.}
\label{fig:mf_plotasymptotic}
\end{figure}
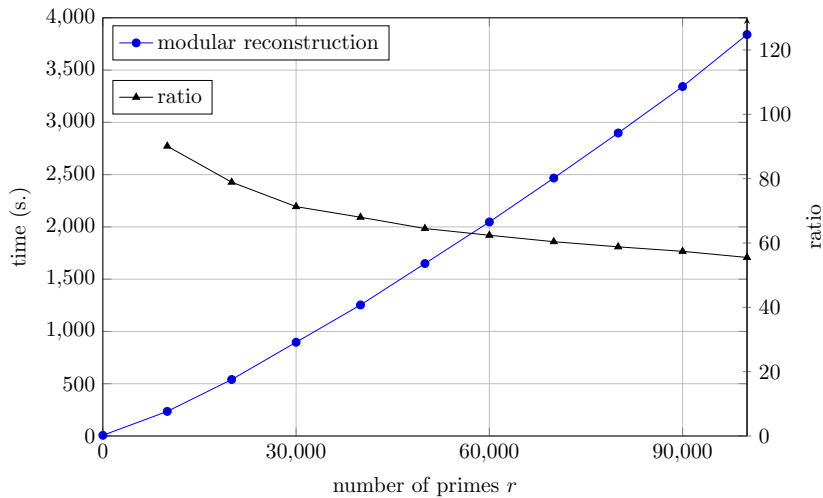

\subsection{Asymptotic Behaviour in the Number of Primes for Geometric Data}
\label{subsec:asymptotics}

A limit of the modular reconstruction algorithm is the arithmetic complexity $\Acpx_r$ for large $r$ ; see Section~\ref{sec:disc_cpx}. Additionally, in the case of topological data analysis, where the underlying space of the sample is unknown, the number $r$ of primes used for multi-field persistence is ``an exploratory parameter'', attempting to find an upper bound $q_r$ on the prime divisor of the torsion coefficients. 

Figures~\ref{fig:mf_plotsmall} and~\ref{fig:mf_plotasymptotic} present the evolution of the running time of the modular reconstruction approach and the brute-force approach for an increasing number of fields $r$ (using the first $r$ prime numbers). Persistence is computed for a Rips complex built on a set of $10\, 000$ points sampling a Klein bottle, which contains torsion in its 
integral homology, resulting in a simplicial complex of $6.14$ million simplices. We analyse the result in terms of the complexity analysis of Section~\ref{sec:mf_cpx}. Here again, $P_\Field$ and $P_r$ remain close during the experiment, even when $r$ grows. The complexity of the brute-force algorithm is $O(r \times P_\Field^3)$ and we indeed observe a linear behaviour when $r$ increases. The complexity of the modular reconstruction approach is $O( \left[ r \times (P_r-P_\Field) + P_r^3 \right] \Acpx_r)$. The part $r \times (P_r-P_\Field)$ of the complexity is negligible because $(P_r-P_\Field)$ is small. For medium values of $r$ ($\leq 150$), like in Figure~\ref{fig:mf_plotsmall}, the arithmetic complexity $O(\Acpx_r)$ increases slowly because $\lambda(Q_{[r]}) = \left\lfloor \log_2 Q_{[r]}/\wordlength \right\rfloor +1$ increases slowly. Together with the little use of arithmetic operations, we consequently observe a very slow increase of the time complexity, compare to the one of brute-force. 

Figure~\ref{fig:mf_plotasymptotic} describes the asymptotic behaviour of the modular approach, where 
the arithmetic operations become costly. We observe that the timings for 
the modular reconstruction approach follow a convex curve. The convexity comes from the growth of 
$\lambda(Q_{[r]})$, which is asymptotically $\Theta(r \log r)$)~\cite{RS62}. However, the increasing of 
the slope is very slow: all along this experiment, 
we have been unable to reach a value of $r$ for which the modular approach is worse than the brute-force approach. For readability, the timings for the brute-force approach are implicitly represented through their ratio with the modular approach: all along the experiment presented in Figure~\ref{fig:mf_plotasymptotic}, for $10\, 000 \leq r \leq 100\, 000$, the modular approach is between $55$ and $90$ times faster. Based on a linear interpolation of the timings for the brute-force approach, and a polynomial interpolation of the modular reconstruction timings, we expect the modular reconstruction to become worse than brute-force for a number of primes $r$ bigger than $4.9$ million. This is due to both the proximity between persistence diagram and multi-field persistence diagram, and the use of only few arithmetic operations by the persistence implementation. 

As a consequence, the modular reconstruction algorithm remains substantially faster than brute force in topological data analysis, for medium to large $r$.

\subsection{Persistence of Torsion in Random Complexes}
\label{subsec:app}

In this section we study the persistence of torsion of Linial-Meshulam random complexes~\cite{Linial*2006}. A Linial-Meshulam 2-complex $Y_2(n,p)$, for an integer $n$ and a probability $0 \leq p \leq 1$, is a random abstract simplicial complex on $n$ vertices made of a complete $1$-skeleton, and where every triangle has been added to $Y(n,p)$ independently with probability $p$. The homology of these complexes have been extensively studied, and they are known to show a short ``burst of torsion'' for certain values of the parameter $p$, with the appearance of a torsion subgroup in homology of experimental super-exponential order~\cite{2017arXiv171005683K,Luczak2018}.

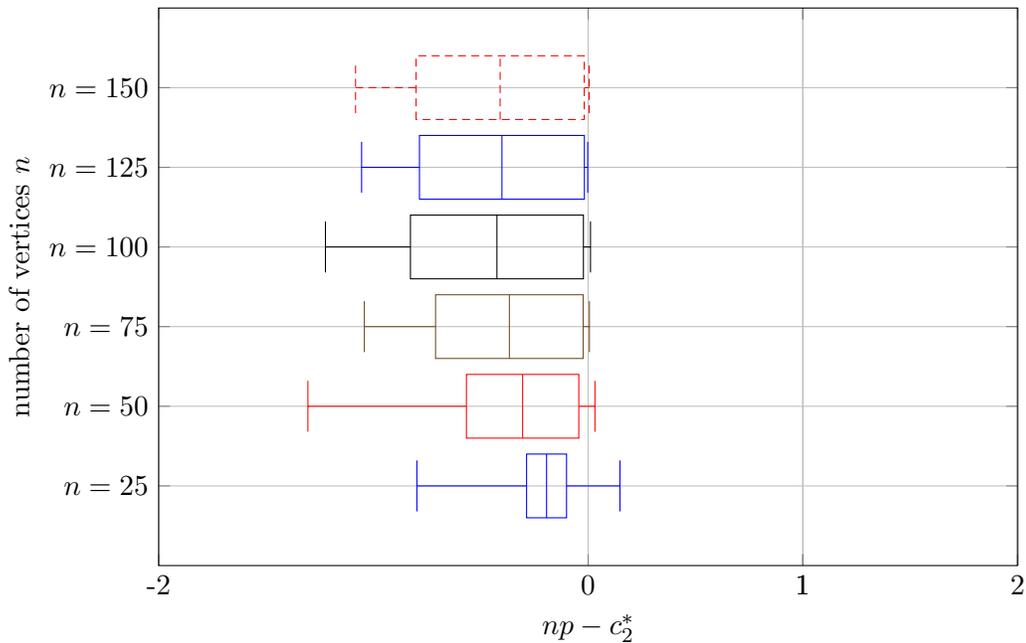
\begin{figure}[t]
\centering
  \begin{tikzpicture}
    \begin{axis} [
      scaled ticks=false,
      width =13cm, 
      height=9cm,
      xmin=-2,
      xmax=2,
      xtick = {-2,1,0,1,2},
      xticklabels = {-2,1,0,1,2},
      grid  = both,
      xlabel = $np -c_2^*$,
      ymin=0,
      ymax=7,
      ytick={1, 2, 3, 4, 5, 6},
      yticklabels={$n=25$, $n=50$, $n=75$, $n=100$, $n=125$, $n=150$},
      ylabel = number of vertices $n$
    ]
      \addplot+ [
        boxplot prepared={
          lower whisker= -0.797478, 
          lower quartile= -0.286609, 
          median= -0.193565,
          upper quartile= -0.100522, 
          upper whisker= 0.148174,
        },
      ] coordinates {}; 
      \addplot+ [
        boxplot prepared={
          lower whisker= -1.30502, 
          lower quartile= -0.566857, 
          median= -0.305122,
          upper quartile= -0.0433878, 
          upper whisker= 0.0317143,
        },
      ] coordinates {}; 
      \addplot+ [
        boxplot prepared={
          lower whisker= -1.04241, 
          lower quartile= -0.710535, 
          median= -0.366817,
          upper quartile= -0.0231003, 
          upper whisker= 0.00608886,
        },
      ] coordinates {}; 
      \addplot+ [
        boxplot prepared={
          lower whisker= -1.22339, 
          lower quartile= -0.827519, 
          median= -0.424736,
          upper quartile= -0.021953, 
          upper whisker= 0.0109969,
        },
      ] coordinates {}; 
      \addplot+ [
        boxplot prepared={
          lower whisker= -1.05494, 
          lower quartile= -0.785487, 
          median= -0.401467,
          upper quartile= -0.0174461, 
          upper whisker= -0.00183635,
        },
      ] coordinates {}; 
      \addplot+ [
        boxplot prepared={
          lower whisker= -1.08313, 
          lower quartile= -0.801343, 
          median= -0.409821,
          upper quartile= -0.0182995, 
          upper whisker= 0.00493343,
        },
      ] coordinates {}; 
    \end{axis}
  \end{tikzpicture}
  \caption{Range $np - c$ for which $H_{1}(Y_2(n,p), \Z)$ admits torsion summands $\Z/q^k \Z$, for $q$ one of the first $200$ prime numbers.}
  \label{fig:stach}
\end{figure}

However, the complete understanding of torsion in the homology of these complexes remains a difficult problem. \L uczak and Peled conjecture the following:
\begin{conjecture}[\L uczak, Peled~\cite{Luczak2018}]
\label{conj:rand}
For $p = p(n)$ such that $|np - c|$ is bounded away from $0$, $H_{1}(Y_2(n,m), \Z)$ is torsion-free asymptotically almost surely, where the constant $c$ is the phase transition constant $c = c_2^*$ of random 2-complexes (see~\cite{Linial2016}[Theorem 1.1]).
\end{conjecture}
In particular, the burst of torsion happens around $p = c_2^* / n$. 

For our experiments, we study the closely related random 2-complex $Y(n,m)$, where $m$ triangles are randomly picked and added to the complex. We use the persistent homology algorithm with torsion to study experimentally the size of the range around $m = c_2^* / n \cdot \binom{n}{3}$ for which $H_{1}(Y_2(n,m), \Z)$ has torsion, for an increasing number of vertices $n$. Our index-valued filtration on $Y(n,m)$ is induced by the random order with which the $m$ triangles are inserted in the complex (all vertices and edges have filtration value $0$). We compute the persistent homology using the modular reconstruction approach for the $r = 200$ first prime numbers.

Figure~\ref{fig:stach} illustrates the intervals of values of $m \in [0, m_{\mathrm{max}}]$ for which the homology of an instance $Y(n,m_{\mathrm{max}})$ contains torsion summands $\Z/q^k \Z$, for $q$ one of the first $200$ prime numbers, and for $n \in \{25,50,75,100,125,150\}$ and $25$ independent runs for each value of $n$. The boxes represent the normalized quantity 
\[
    n \cdot m \cdot \binom{n}{3}^{-1} - c_2^*
\]
to correspond to Conjecture~\ref{conj:rand}. The boxes represent the average lower bound and upper bound (and centre) of the intervals, and the whiskers stand for the extremal values observed in the samples.

Similarly to the study of homology with field coefficients~\cite{Linial2016}, we observe a one-sided sharp transition at $p = c_2^* / n$ for the disappearance of torsion. The plot seems to corroborate the convergence of a lower bound for the interval at a constant $k_2 \approx -0.8$, which suggests that, following Conjecture~\ref{conj:rand}, the homology group $H_{1}(Y_2(n,m), \Z)$ is torsion-free a.a.s. when $|np - c| > k_2 \approx -0.8$.

\section*{Acknowledgement}

The research leading to these results has received funding from the European Research Council (ERC) under  the  European  Union's  Seventh  Framework  Programme (FP/2007-2013)  /  ERC  Grant Agreement No. 339025 GUDHI (Algorithmic Foundations of Geometry Understanding in Higher Dimensions).

\subsubsection*{Conflict of interest.} On behalf of all authors, the corresponding author states that there is no conflict of interest. 

\bibliographystyle{plain}
\bibliography{bibliography}

\appendix

\section{Arithmetic Notations}

\begin{itemize}
\item $\Z$ ring of integers,
\item $\Z/n\Z$ ring of integers modulo $n \geq 2$,
\item $\Q$ field of rationals,
\item $q_1, \ldots, q_r$ the $r$ first prime numbers, for $r \geq 1$,
\item $[r]$ the set $\{1, \ldots , r\}$,
\item $Q := q_1 \times \ldots \times q_r$, product of first $r$ prime numbers,
\item $Q_S := \prod_{s \in S} q_s$, for a subset of indices $S \subset [r]$,
\item indices $s$, $t$, $r$, and set of indices $S$ and $T$, are reserved to the indexing of prime numbers $\{q_1, \ldots, q_r\}$,
\item indices $i$, $j$, $k$ and $m$ refer to indices in the filtration of a complex, and hence indices for matrix columns and matrix reduction algorithms. 
\end{itemize}

\end{document}